\documentclass[12pt]{article}    

\usepackage[margin=0.75in]{geometry} 
\usepackage{amsmath,amssymb,amsthm,upgreek}

\newtheorem{theorem}{Theorem}[section]

\newtheorem{lemma}[theorem]{Lemma}
\newtheorem{corollary}[theorem]{Corollary}

\newtheorem{remark}{Remark}

\newcommand{\al}{\alpha}
\newcommand{\bt}{\beta}
\newcommand{\la}{\lambda}
\newcommand{\ka}{\kappa}
\newcommand{\s}{\sigma}
\newcommand{\be}{\begin{equation}}
\newcommand{\ee}{\end{equation}}
\newcommand{\no}{\nonumber}

\numberwithin{equation}{section}
\begin{document}

\title{Generalized Airy polynomials, Hankel determinants and asymptotics}
\author{Chao Min\thanks{School of Mathematical Sciences, Huaqiao University, Quanzhou 362021, China; Email: chaomin@hqu.edu.cn}\: and Pixin Fang\thanks{School of Mathematical Sciences, Huaqiao University, Quanzhou 362021, China}}


\date{\today}
\maketitle
\begin{abstract}
We further study the orthogonal polynomials with respect to the generalized Airy weight based on the work of Clarkson and Jordaan [{\em J. Phys. A: Math. Theor.} {\bf 54} ({2021}) {185202}]. We prove the ladder operator equations and associated compatibility conditions for orthogonal polynomials with respect to a general Laguerre-type weight of the form $w(x)=x^\lambda w_0(x),\;\la>-1, x\in\mathbb{R}^+$. By applying them to the generalized Airy polynomials, we are able to derive a discrete system for the recurrence coefficients. Combining with the Toda evolution, we establish the relation between the recurrence coefficients, the sub-leading coefficient of the monic generalized Airy polynomials and the associated Hankel determinant. Using Dyson's Coulomb fluid approach and with the aid of the discrete system for the recurrence coefficients, we obtain the large $n$ asymptotic expansions for the recurrence coefficients and the sub-leading coefficient of the monic generalized Airy polynomials. The large $n$ asymptotic expansion (including the constant term) of the Hankel determinant has been derived by using a recent result in the literature. The long-time asymptotics of these quantities have also been discussed explicitly.
\end{abstract}

$\mathbf{Keywords}$: Generalized Airy polynomials; Ladder operators; Recurrence coefficients;

Hankel determinants; Coulomb fluid; Asymptotic expansions.

$\mathbf{Mathematics\:\: Subject\:\: Classification\:\: 2020}$: 42C05, 33C45, 41A60.

\section{Introduction}
It is a well-known fact that the weight functions $w(x)$ of classical orthogonal polynomials (such as Hermite, Laguerre and Jacobi polynomials) satisfy the Pearson equation
\be\label{pe}
\frac{d}{dx}(\sigma(x)w(x))=\tau(x)w(x),
\ee
where $\sigma(x)$ and $\tau(x)$ are polynomials with $\mathrm{deg}(\sigma(x))\le2$ and $\mathrm{deg}(\tau(x))=1$.
Semi-classical orthogonal polynomials have the weights $w(x)$ that satisfy the equation (\ref{pe}) where $\sigma(x)$ and $\tau(x)$ are polynomials with either $\mathrm{deg}(\sigma(x))>2$ or $\mathrm{deg}(\tau(x))\neq 1$. See, e.g., \cite[Section 1.1.1]{VanAssche}.

Recently, Clarkson and Jordaan {\cite{Clarkson4}} considered the semi-classical orthogonal polynomials with respect to the so-called generalized Airy weight
\be\label{airy}
w(x)=w(x;t)=x^{\lambda}\mathrm{e}^{-\frac{1}{3}x^3+tx},\qquad x\in\mathbb{R}^+
\ee
with parameters $\lambda>-1,\; t\in\mathbb{R}$. The moments of the weight (\ref{airy}) can be expressed in terms of the Airy functions for some particular $\la$. They mainly derived the differential and difference equations for the generalized Airy polynomials and for the recurrence coefficients. They also studied properties of the zeros of the polynomials and various asymptotics of the recurrence coefficients.

The weight function (\ref{airy}) is indeed a semi-classical weight, since it satisfies the Pearson equation (\ref{pe}) with
$$
\sigma(x)=x,\qquad \tau(x)=-x^3+tx+\lambda+1.
$$
Orthogonal polynomials associated with the exponential cubic weight and its variants have also been studied in \cite{Magnus,MS,VFZ} and have important applications in
numerical analysis \cite{DHK} and random matrix theory known as the cubic matrix model \cite{BD,BD1,BDY}.

Let $P_n(x),\; n=0, 1, 2, \ldots$ be the \textit{monic} polynomials of degree $n$, orthogonal with respect to the weight function (\ref{airy}), namely,
\be\label{orthre}
\int_0^{\infty}P_m(x)P_n(x)w(x)dx=h_n\delta_{mn},\qquad m,n=0,1,2,\dots,
\ee
where $\delta_{mn}$ is the Kronecker delta, $h_n>0$ is the normalized constant and $P_n(x)$ has the monomial expansion
\be\label{pnxt1}
P_n(x)=x^n+{\rm\bf p}(n)x^{n-1}+\cdots+P_n(0),
\ee
where ${\rm\bf p}(n)$, the sub-leading coefficient of $P_n(x)$, plays an important role in the following discussions.

One of the most important properties of orthogonal polynomials is that
they satisfy the three-term recurrence relation of the form
\be\label{pnxt}
xP_n(x)=P_{n+1}(x)+\alpha_nP_n(x)+\beta_nP_{n-1}(x),
\ee
with the initial conditions
$$
P_0(x)=1,\qquad \beta_0P_{-1}(x)=0.
$$

Due to the dependence on the parameter $t$ in the weight function (\ref{airy}), all the quantities like $P_n(x),\; h_n,\; {\rm\bf p}(n)$ and the recurrence coefficients $\al_n,\; \bt_n$ are actually functions of $t$.

It is easy to see that the recurrence coefficients $\al_n$ and $\bt_n$ have the following expressions:
\be\label{alpha}
\alpha_n={\rm\bf p}(n)-{\rm\bf p}(n+1),
\ee
\be\label{beta}
\beta_n=\frac{h_n}{h_{n-1}}.
\ee
A telescopic sum of (\ref{alpha}) produces
\be\label{suma}
\sum_{j=0}^{n-1}\alpha_j=-{\rm\bf p}(n),
\ee
where we have used the initial condition ${\rm\bf p}(0)=0$.
Furthermore, from (\ref{pnxt}) we have the Christoffel-Darboux formula
\be\label{cd}
\sum_{j=0}^{n-1}\frac{P_j(x)P_j(y)}{h_j}=\frac{P_n(x)P_{n-1}(y)-P_n(y)P_{n-1}(x)}{h_{n-1}(x-y)}.
\ee
For more information about orthogonal polynomials, see \cite{Chihara,Ismail,Szego}.

The Hankel determinant generated by the weight (\ref{airy}) is
$$
D_{n}(t):=\det\big(\mu_{i+j}(t)\big)_{i,j=0}^{n-1}=\begin{vmatrix}
\mu_{0}(t)&\mu_{1}(t)&\cdots&\mu_{n-1}(t)\\
\mu_{1}(t)&\mu_{2}(t)&\cdots&\mu_{n}(t)\\
\vdots&\vdots&&\vdots\\
\mu_{n-1}(t)&\mu_{n}(t)&\cdots&\mu_{2n-2}(t)
\end{vmatrix},\\
$$
where $\mu_j(t)$ is the $j$th moment given by
$$
\mu_j(t):=\int_{0}^{\infty}x^jw(x;t)dx,\qquad j=0,1,2,\dots
$$
and it can be expressed in terms of the generalized hypergeometric functions \cite{Clarkson4}.
The Hankel determinant $D_n(t)$ is equal to the product of $h_j(t)$ in the form {\cite[(2.1.6)]{Ismail}}
\be\label{dnt1}
D_n(t)=\prod_{j=0}^{n-1}h_j(t).
\ee
From (\ref{beta}) and (\ref{dnt1}), we have the relation between $\beta_n(t)$ and $D_n(t)$:
\be\label{betadn}
\beta_n(t)=\frac{D_{n+1}(t)D_{n-1}(t)}{D_n^2(t)}.
\ee


The remainder of this paper is organized as follows. In Section 2, we prove the ladder operator equations and associated compatibility conditions for orthogonal polynomials with respect to a general Laguerre-type weight. In Section 3, we apply the ladder operators and compatibility conditions to the generalized Airy polynomials. This enables us to derive a discrete system for the recurrence coefficients, and establish the relations between the logarithmic derivative of the Hankel determinant, the sub-leading coefficient ${\rm\bf p}(n,t)$ and the recurrence coefficients. In Section 4, we study the large $n$ asymptotics of the recurrence coefficients, the sub-leading coefficient ${\rm\bf p}(n,t)$, the Hankel determinant $D_n(t)$ and the normalized constant $h_n(t)$ for fixed $t\in\mathbb{R}$ by using Dyson's Coulomb fluid approach. The long-time asymptotics ($t\rightarrow\pm\infty$) of these quantities for fixed $n\in\mathbb{N}$ have been investigated in Section 5. Finally, we present some concluding remarks in Section 6.

\section{Ladder operators and compatibility conditions}
The ladder operators and compatibility conditions for orthogonal polynomials have been derived by Chen and Ismail \cite{ChenIsmail2,ChenIsmail2005}. However, they require that the weight be well defined at the endpoints of the interval (vanish at the endpoints most of the time in practice). This is not the case for the generalized Airy weight since it is undefined at $0$ when $\la<0$.

To solve this problem,
in this section we derive the ladder operator equations and compatibility conditions for the monic orthogonal polynomials with respect to a general Laguerre-type weight of the form
\be\label{lw}
w(x)=x^\la w_0(x),\qquad x\in\mathbb{R}^+,
\ee
where $\lambda>-1,\; w_0(x)$ is a continuously differentiable function defined on $[0,\infty)$ and all the moments
$
\int_{0}^{\infty}x^jw(x)dx,\; j=0,1,2,\dots
$
exist. We also require that the weight $w(x)$ be rapidly decreasing in the sense that $\lim\limits_{x\rightarrow+\infty}\pi(x)w(x)=0$ for an arbitrary polynomial $\pi(x)$. Typical examples for such weights are\\
1. the classical Laguerre weight $w(x)=x^{\lambda}\mathrm{e}^{-x},\; x\in\mathbb{R}^+,\; \la>-1$,\\
2. the semi-classical Laguerre weight $w(x)=x^{\lambda}\mathrm{e}^{-x^2+tx},\; x\in\mathbb{R}^+,\; \la>-1,\; t\in\mathbb{R}$ \cite{Boelen,Clarkson,Filipuk,Min2023},\\
3. the generalized Airy weight (\ref{airy}).\\
Obviously, the formulas (\ref{orthre})--(\ref{cd}) still hold for the general Laguerre-type weight (\ref{lw}).

\begin{theorem}\label{low 1}
The monic orthogonal polynomials with respect to the general Laguerre-type weight (\ref{lw}) satisfy the lowering operator equation
\be\label{lower}
\left(\frac{d}{dx}+B_n(x)\right)P_n(x)=\beta_nA_n(x)P_{n-1}(x),
\ee
where $A_n(x)$ and $B_n(x)$ are given by
\begin{subequations}\label{anbn1}
\be
A_n(x):=\frac{1}{x}\cdot\frac{1}{h_n}\int_0^{\infty}\frac{x\mathrm{v}'(x)-y\mathrm{v}'(y)}{x-y}P_n^2(y)w(y)dy,\label{eqn:aa}
\ee
\be
B_n(x):=\frac{1}{x}\left(\frac{1}{h_{n-1}}\int_0^{\infty}\frac{x\mathrm{v}'(x)-y\mathrm{v}'(y)}{x-y}P_n(y)P_{n-1}(y)w(y)dy-n\right),\label{eqn:ab}
\ee
\end{subequations}
and $\mathrm{v}(x)=-\ln w(x)$.
\end{theorem}

\begin{proof}
Since $P_n(x)$ is a polynomial of degree $n$, we have
\be\label{xpn}
xP_n'(x)=\sum_{k=0}^n c_{nk}P_k(x).
\ee
Using the orthogonality relation (\ref{orthre}), we find
$$
c_{nk}=\frac{1}{h_k}\int_0^{\infty}y P_n'(y)P_k(y)w(y)dy.
$$
When $k=n$, it is obvious to see that
\be\label{cnn}
c_{nn}=n.
\ee
When $k=0, 1, \ldots, n-1$, we have through integration by parts
\begin{align}\label{hkcn}
h_k c_{nk}&=\int_0^{\infty}yP_k(y)w(y)dP_n(y) \no\\
&=\Big[yP_k(y)P_n(y)w(y)\Big]_0^{\infty}-\int_0^{\infty}\left(yP_k(y)w(y)\right)'P_n(y)dy\no\\
&=-\int_0^{\infty}\left(P_n(y)P_k(y)+yP_n(y)P_k'(y)\right)w(y)dy-\int_0^{\infty}yP_n(y)P_k(y)w'(y)dy\no\\
&=-\int_0^{\infty}yP_n(y)P_k(y)w'(y)dy,\no
\end{align}
i.e.,
\be\label{cnk}
c_{nk}=-\frac{1}{h_k}\int_0^{\infty}yP_n(y)P_k(y)w'(y)dy,\qquad k=0, 1, \ldots, n-1.
\ee
Substituting (\ref{cnn}) and (\ref{cnk}) into (\ref{xpn}), we find
\begin{align}
xP_n'(x)&=nP_n(x)-\sum_{k=0}^{n-1}\frac{P_k(x)}{h_k}\int_0^{\infty}yP_n(y)P_k(y)w'(y)dy\no\\
&=nP_n(x)-\sum_{k=0}^{n-1}\frac{P_k(x)}{h_k}\int_0^{\infty}yP_n(y)P_k(y)(-\mathrm{v}'(y))w(y)dy\no\\
&=nP_n(x)-\sum_{k=0}^{n-1}\frac{P_k(x)}{h_k}\int_0^{\infty}P_n(y)P_k(y)(x\mathrm{v}'(x)-y\mathrm{v}'(y))w(y)dy\no\\
&=nP_n(x)-\int_0^{\infty}P_n(y)\left(\sum_{k=0}^{n-1}\frac{P_k(x)P_k(y)}{h_k}\right)(x\mathrm{v}'(x)-y\mathrm{v}'(y))w(y)dy.\no
\end{align}
By making use of the Christoffel-Darboux formula (\ref{cd}),
it follows that
\begin{align}
xP_n'(x)
=&-\left(\frac{1}{h_{n-1}}\int_0^{\infty}\frac{x\mathrm{v}'(x)-y\mathrm{v}'(y)}{x-y}P_n(y)P_{n-1}(y)w(y)dy-n\right)P_n(x)\no\\
&+\left(\frac{1}{h_{n-1}}\int_0^{\infty}\frac{x\mathrm{v}'(x)-y\mathrm{v}'(y)}{x-y}P_n^2(y)w(y)dy\right)P_{n-1}(x),\no
\end{align}
which is equivalent to (\ref{lower}) with $A_n(x)$ and $B_n(x)$ given in (\ref{anbn1}).
\end{proof}
\begin{remark}
When $\la>0$, it is easy to find that the definitions of $A_n(x)$ and $B_n(x)$ in (\ref{anbn1}) can be simplified to
$$
A_n(x)=\frac{1}{h_n}\int_0^{\infty}\frac{\mathrm{v}'(x)-\mathrm{v}'(y)}{x-y}P_n^2(y)w(y)dy,
$$
$$
B_n(x)=\frac{1}{h_{n-1}}\int_0^{\infty}\frac{\mathrm{v}'(x)-\mathrm{v}'(y)}{x-y}P_n(y)P_{n-1}(y)w(y)dy,
$$
which is the case of Chen and Ismail \cite{ChenIsmail2,ChenIsmail2005}. Hence, it should also be pointed out that Theorem 3.10 and Lemma 3.11 in \cite{Clarkson4} correspond to the case $\la>0\; (w(0)=w(\infty)=0)$. See also Remark \ref{re} in the next section.
\end{remark}
\begin{remark}
The formulas in (\ref{anbn1}) are very effective. Consider a simple example for the monic Laguerre polynomials with the weight $w(x)=\mathrm{e}^{-x},\;x\in\mathbb{R}^+$. It is easy to see from (\ref{anbn1}) that $A_n(x)=1/x,\;B_n(x)=-n/x$, which is more straightforward than \cite[(1.5) and (1.6)]{ChenIsmail2005}.
\end{remark}
By making use of the definitions of $A_n(x)$ and $B_n(x)$ in (\ref{anbn1}), we will prove that the compatibility conditions (\ref{s1}), (\ref{s2}) and (\ref{s2'}) still hold in the following theorems.
\begin{theorem}
The functions $A_n(x)$ and $B_n(x)$ satisfy the first supplementary condition
\be\label{s1}
B_{n+1}(x)+B_n(x)=(x-\alpha_n)A_n(x)-{\rm v}'(x).\tag{$S_1$}
\ee
\end{theorem}
\begin{proof}
From (\ref{eqn:ab}) we have
$$
B_{n+1}(x)+B_n(x)=\frac{1}{x}\int_0^{\infty}\frac{x\mathrm{v}'(x)-y\mathrm{v}'(y)}{x-y}\left(\frac{P_{n+1}(y)}{h_n}+\frac{P_{n-1}(y)}{h_{n-1}}\right)P_n(y)w(y)dy-\frac{2n+1}{x}.
$$
Using the three-term recurrence relation (\ref{pnxt}) yields
$$
\frac{P_{n+1}(y)}{h_n}+\frac{P_{n-1}(y)}{h_{n-1}}=\frac{y-\alpha_n}{h_n}P_n(y).
$$
It follows that
$$
B_{n+1}(x)+B_n(x)=\frac{1}{x}\cdot\frac{1}{h_n}\int_0^{\infty}\frac{x\mathrm{v}'(x)-y\mathrm{v}'(y)}{x-y}(y-\alpha_n)P_n^2(y)w(y)dy-\frac{2n+1}{x}.
$$
Using (\ref{eqn:aa}) we find
\begin{align}\label{bnm}
B_{n+1}(x)+B_n(x)-(x-\alpha_n)A_n(x)
&=-\frac{1}{x}\cdot\frac{1}{h_n}\int_0^{\infty}(x\mathrm{v}'(x)-y\mathrm{v}'(y))P_n^2(y)w(y)dy-\frac{2n+1}{x}\no\\
&=-\mathrm{v}'(x)+\frac{1}{x}\cdot\frac{1}{h_n}\int_0^{\infty}y\mathrm{v}'(y)P_n^2(y)w(y)dy-\frac{2n+1}{x}.
\end{align}
Through integration by parts,
\begin{align}\label{in}
\int_0^{\infty}y\mathrm{v}'(y)P_n^2(y)w(y)dy&=-\int_0^{\infty}yP_n^2(y)dw(y)\no\\
&=-\Big[yw(y)P_n^2(y)\Big]_0^{\infty}+\int_0^{\infty}\left(yP_n^2(y)\right)'w(y)dy\no\\
&=\int_0^{\infty}P_n^2(y)w(y)dy+\int_0^{\infty}2yP_n(y)P_n'(y)w(y)dy\no\\
&=(2n+1)h_n.
\end{align}
Substituting (\ref{in}) into (\ref{bnm}), we establish the theorem.
\end{proof}
Before we prove the compatibility conditions (\ref{s2}) and (\ref{s2'}), we show that
the combination of the lowering operator equation (\ref{lower}) and the condition (\ref{s1}) produces the raising operator equation.
\begin{theorem}
The monic orthogonal polynomials $P_n(x)$ satisfy the raising operator equation
\be\label{ra}
\left(\frac{d}{dx}-B_n(x)-{\rm v}'(x)\right)P_{n-1}(x)=-A_{n-1}(x)P_n(x).
\ee
\end{theorem}
\begin{proof} From  (\ref{lower}) we have
$$
P_{n-1}'(x)=-B_{n-1}(x)P_{n-1}(x)+\beta_{n-1}A_{n-1}(x)P_{n-2}(x).
$$
Using the three-term recurrence relation and replacing $\beta_{n-1}P_{n-2}(x)$ by $(x-\alpha_{n-1})P_{n-1}(x)-P_n(x)$, it follows that
$$
P_{n-1}'(x)
=\left[(x-\alpha_{n-1})A_{n-1}(x)-B_{n-1}(x)\right]P_{n-1}(x)-A_{n-1}(x)P_n(x).
$$
By making use of (\ref{s1}), we obtain
$$
P_{n-1}'(x)=(B_n(x)+\mathrm{v}'(x))P_{n-1}(x)-A_{n-1}(x)P_n(x).
$$
The proof is complete.
\end{proof}
\begin{corollary}
The monic orthogonal polynomials $P_n(x)$ satisfy the second-order differential equation
\begin{align}\label{ode}
&P_n''(x)-\bigg(\mathrm{v}'(x)+\frac{A_n'(x)}{A_n(x)}\bigg)P_n'(x)+\bigg(B_n'(x)-B_n^2(x)-\mathrm{v}'(x)B_n(x)\no\\[8pt]
&+\beta_nA_n(x)A_{n-1}(x)-\frac{A_n'(x)B_n(x)}{A_n(x)}\bigg)P_n(x)=0.
\end{align}
\end{corollary}
\begin{proof}
Eliminating $P_{n-1}(x)$ from the ladder operator equations (\ref{lower}) and (\ref{ra}) gives the desired result.
\end{proof}
\begin{theorem}
The functions $A_n(x)$ and $B_n(x)$ satisfy the second supplementary condition
\be\label{s2}
1+(x-\alpha_n)(B_{n+1}(x)-B_n(x))=\beta_{n+1}A_{n+1}(x)-\beta_nA_{n-1}(x).\tag{$S_2$}
\ee
\end{theorem}
\begin{proof}
Using the definition of $B_n(x)$ in (\ref{eqn:ab}), we have
\begin{align}
&(x-\alpha_n)(B_{n+1}(x)-B_n(x))\no\\
&=(x-\alpha_n)\cdot\frac{1}{x}\int_0^{\infty}\frac{x\mathrm{v}'(x)-y\mathrm{v}'(y)}{x-y}\left(\frac{P_{n+1}(y)}{h_n}-\frac{P_{n-1}(y)}{h_{n-1}}\right)P_n(y)w(y)dy-\frac{x-\alpha_n}{x}
\no\\
&=\frac{1}{x}\int_0^{\infty}[(x-y)+(y-\alpha_n)]\frac{x\mathrm{v}'(x)-y\mathrm{v}'(y)}{x-y}\left(\frac{P_{n+1}(y)}{h_n}-\frac{P_{n-1}(y)}{h_{n-1}}\right)P_n(y)w(y)dy
-\frac{x-\alpha_n}{x}\no\\
&=\frac{1}{x}\int_0^{\infty}(x\mathrm{v}'(x)-y\mathrm{v}'(y))\left(\frac{P_{n+1}(y)}{h_n}-\frac{P_{n-1}(y)}{h_{n-1}}\right)P_n(y)w(y)dy\no\\
&\quad+\frac{1}{x}\cdot\frac{1}{h_n}\int_0^{\infty}(y-\alpha_n)P_n(y)(P_{n+1}(y)-\bt_nP_{n-1}(y))\frac{x\mathrm{v}'(x)-y\mathrm{v}'(y)}{x-y}w(y)dy-\frac{x-\alpha_n}{x}.\no
\end{align}
Replacing $(y-\alpha_n)P_n(y)$ by $P_{n+1}(y)+\bt_nP_{n-1}(y)$ yields
\begin{align}\label{eq}
&(x-\alpha_n)(B_{n+1}(x)-B_n(x))\no\\
&=\frac{1}{x}\cdot\frac{1}{h_{n-1}}\int_0^{\infty}yP_n(y)P_{n-1}(y)\mathrm{v}'(y)w(y)dy-\frac{1}{x}\cdot\frac{1}{h_n}\int_0^{\infty}y P_{n+1}(y)P_n(y)\mathrm{v}'(y)w(y)dy\no\\
&\quad+\frac{1}{x}\cdot\frac{1}{h_n}\int_0^{\infty}\left(P_{n+1}^2(y)-\bt_n^2P_{n-1}^2(y)\right)\frac{x\mathrm{v}'(x)-y\mathrm{v}'(y)}{x-y}w(y)dy-\frac{x-\alpha_n}{x}.
\end{align}
Through integration by parts, we find
\begin{align}
\frac{1}{h_{n-1}}\int_0^{\infty}yP_n(y)P_{n-1}(y)\mathrm{v}'(y)w(y)dy&=-\frac{1}{h_{n-1}}\int_0^{\infty}yP_n(y)P_{n-1}(y)dw(y)\no\\
&=\frac{1}{h_{n-1}}\int_0^{\infty}yP_n'(y)P_{n-1}(y)w(y)dy\no\\
&=-{\rm\bf p}(n),\no
\end{align}
where use has been made of the fact that
$$yP_n'(y)=nP_n(y)-{\rm\bf p}(n)P_{n-1}(y)+\mathrm{lower}\; \mathrm{degree}\; \mathrm{polynomials}.$$
In view of (\ref{alpha}), it follows from (\ref{eq}) that
\begin{align}
&1+(x-\alpha_n)(B_{n+1}(x)-B_n(x))\no\\
&=\frac{1}{x}\cdot\frac{1}{h_n}\int_0^{\infty}\left(P_{n+1}^2(y)-\bt_n^2P_{n-1}^2(y)\right)\frac{x\mathrm{v}'(x)-y\mathrm{v}'(y)}{x-y}w(y)dy\no\\
&=\beta_{n+1}A_{n+1}(x)-\beta_nA_{n-1}(x).\no
\end{align}
This completes the proof.
\end{proof}
The combination of (\ref{s1}) and (\ref{s2}) produces the sum rule (\ref{s2'}) in the following theorem.
\begin{theorem}
The functions $A_n(x)$ and $B_n(x)$ satisfy the condition
\be\label{s2'}
B_n^2(x)+{\rm v}'(x)B_n(x)+\sum_{j=0}^{n-1}A_j(x)=\beta_nA_n(x)A_{n-1}(x). \tag{$S_2'$}
\ee
\end{theorem}
\begin{proof}
Multiplying by $A_n(x)$ on both sides of  (\ref{s2}), we have
$$
A_n(x)+(x-\alpha_n)A_n(x)(B_{n+1}(x)-B_n(x))
=\beta_{n+1}A_{n+1}(x)A_n(x)-\beta_nA_n(x)A_{n-1}(x).
$$
Using (\ref{s1}) to replace $(x-\alpha_n)A_n(x)$ by $B_{n+1}(x)+B_n(x)+\mathrm{v}'(x)$ gives
$$
A_n(x)+B_{n+1}^2(x)-B_n^2(x)+\mathrm{v}'(x)(B_{n+1}(x)-B_n(x))
=\beta_{n+1}A_{n+1}(x)A_n(x)-\beta_nA_n(x)A_{n-1}(x).
$$
Taking a telescopic sum, together with the initial conditions $B_0(x)=0,\; \beta_0A_{-1}(x)=0$, we establish the theorem.
\end{proof}
\begin{remark}
Using (\ref{s2'}), the differential equation (\ref{ode}) can be written in the form
$$
P_n''(x)-\left(\mathrm{v}'(x)+\frac{A_{n}'(x)}{A_{n}(x)}\right)P_n'(x)+\left(B_{n}'(x)-B_{n}(x)\frac{A_{n}'(x)}{A_{n}(x)}
+\sum_{j=0}^{n-1}A_{j}(x)\right)P_n(x)=0.
$$
\end{remark}
\section{Generalized Airy polynomials and Hankel determinants}
In this section, we apply the ladder operators and compatibility conditions to the generalized Airy polynomials.
Recall that the generalized Airy weight reads
$$
w(x)=x^{\lambda}\mathrm{e}^{-\frac{1}{3}x^3+tx},\qquad x\in\mathbb{R}^+,
$$
with parameters $\lambda>-1,\; t\in\mathbb{R}$. It follows that
\be\label{po}
\mathrm{v}(x)=\frac{1}{3}x^3-tx-\lambda\ln x
\ee
and
\be\label{vd}
\frac{x\mathrm{v}'(x)-y\mathrm{v}'(y)}{x-y}=x^2+xy+y^2-t.
\ee

Inserting (\ref{vd}) into the definitions of $A_n(x)$ and $B_n(x)$ in (\ref{anbn1}), we obtain
\be\label{anbn}
A_n(x)=x+\alpha_n+\frac{R_n}{x},\qquad B_n(x)=\beta_n+\frac{r_n}{x},
\ee
where
$$
R_n=\frac{1}{h_n}\int_{0}^{\infty}y^2P_n^2(y)w(y)dy-t,
$$
$$
r_n=\frac{1}{h_{n-1}}\int_{0}^{\infty}y^2P_n(y)P_{n-1}(y)w(y)dy-n.
$$
\begin{remark}\label{re}
It is easy to see that
$$
\frac{\mathrm{v}'(x)-\mathrm{v}'(y)}{x-y}=x+y+\frac{\la}{xy}.
$$
This leads to the fact that the expressions of $A_n(x)$ and $B_n(x)$ in \cite[Theorem 3.10]{Clarkson4} are undefined for $\la<0$.
\end{remark}
Substituting (\ref{anbn}) into (\ref{s1}), we obtain two identities by equating the powers of $x$ on both sides:
\be\label{s11}
\bt_{n+1}+\bt_n=R_n -\al_n^2 +t,
\ee
\be\label{s12}
r_{n+1}+r_n=\lambda-\al_nR_n.
\ee
Similarly, substituting (\ref{anbn}) into (\ref{s2'}) produces the following four identities:
\be\label{s21}
r_n+n=\bt_n(\al_n+\al_{n-1}),
\ee
\be\label{s22}
r_n^2-\lambda r_n=\bt_nR_nR_{n-1},
\ee
\be\label{s23}
\bt_n^2-t\bt_n+\sum_{j=0}^{n-1}\al_j=\bt_n(\al_n\al_{n-1}+R_n+R_{n-1}),
\ee
$$
2\bt_nr_n-\lambda\bt_n-tr_n+\sum_{j=0}^{n-1}R_j=\bt_n(\al_nR_{n-1}+\al_{n-1}R_n).
$$
From (\ref{s11}) and (\ref{s21}), we can express the auxiliary quantities $R_n$ and $r_n$ in terms of the recurrence coefficients:
\begin{subequations}\label{Rnrn}
\be\label{Rn}
R_n=\alpha_n^2+\beta_n+\beta_{n+1}-t,
\ee
\be
r_n=(\alpha_n+\alpha_{n-1})\beta_n-n.
\ee
\end{subequations}
\begin{theorem}
The recurrence coefficients $\alpha_n$ and $\beta_n$ for the generalized Airy polynomials satisfy the discrete system
\begin{subequations}\label{dieq}
\be
\alpha_n^3-t\alpha_n+(2\alpha_n+\alpha_{n-1})\beta_n+(2\alpha_n+\alpha_{n+1})\beta_{n+1}=2n+\lambda+1,\label{dieq1}
\ee
\be
[(\alpha_n+\alpha_{n-1})\beta_n-n]^2-\la[(\alpha_n+\alpha_{n-1})\beta_n-n]=\bt_n(\alpha_n^2+\beta_n+\beta_{n+1}-t)(\alpha_{n-1}^2+\beta_{n-1}+\beta_{n}-t).
\ee
\end{subequations}
\end{theorem}
\begin{proof}
Substituting (\ref{Rnrn}) into (\ref{s12}) and (\ref{s22}) respectively, we establish the theorem.
See also \cite[Theorem 3.12]{Clarkson4}.
\end{proof}

\begin{theorem} The sub-leading coefficient ${\rm\bf p}(n)$ can be expressed in terms of the recurrence coefficients as follows:
\be\label{pnt}
{\rm\bf p}(n)=-\beta_n(\alpha_{n-1}^2+\alpha_{n-1}\alpha_n+\alpha_n^2+\beta_{n-1}+\beta_n+\beta_{n+1}-t).
\ee
\end{theorem}
\begin{proof}
Combining (\ref{suma}) and (\ref{s23}) yields
$$
{\rm\bf p}(n)=-\beta_n(\alpha_n\alpha_{n-1}+ R_n +R_{n-1}-\beta_n+t).
$$
Substituting (\ref{Rn}) into the above gives the desired result.
\end{proof}

Let $H_n(t)$ be the logarithmic derivative of the Hankel determinant, namely,
\be\label{Hntd}
H_n(t):=\frac{d}{dt}\ln D_n(t).
\ee
Next, we would like to derive the expression of $H_n(t)$ in terms of the recurrence coefficients.
From the orthogonality condition (\ref{orthre}), we have
$$
\int_0^{\infty}P_n^2(x;t)w(x;t)dx= h_n(t)
$$
and
$$
\int_0^\infty P_n(x;t)P_{n-1}(x;t)w(x;t)dx=0.
$$
Taking derivatives with respect to $t$, we obtain
\be\label{hna}
\frac{d}{dt}\ln h_n(t)=\alpha_n
\ee
and
\be\label{dpn}
\frac{d}{dt}{\rm\bf p}(n,t)=-\beta_n,
\ee
respectively.
\begin{remark}
It is easy to show that the recurrence coefficients satisfy the Toda system from (\ref{alpha}), (\ref{beta}), (\ref{hna}) and (\ref{dpn}):
\begin{subequations}\label{ts}
\be
\frac{d\alpha_n}{dt}=\beta_{n+1}-\beta_n,
\ee
\be
\frac{d\beta_n}{dt}=\beta_n(\alpha_n-\alpha_{n-1}).
\ee
\end{subequations}
In fact, the Toda system (\ref{ts}) hold for a very general class of orthogonal polynomials with the weight of the form $w(x;t)=\widetilde{w}_0(x)\mathrm{e}^{tx}$ and the weight has finite moments for all $t\in\mathbb{R}$; see \cite{ChenIsmail} and also \cite[Section 2.8]{Ismail}.
\end{remark}
\begin{theorem}
The logarithmic derivative of the Hankel determinant, $H_n(t)$, has two alternative representations as follows:
\begin{align}
H_n(t)
&=-{\rm\bf p}(n,t)\label{Hnt1}\\
&=\beta_n(\alpha_{n-1}^2+\alpha_{n-1}\alpha_n+\alpha_n^2+\beta_{n-1}+\beta_n+\beta_{n+1}-t).\label{Hnt2}
\end{align}
\end{theorem}
\begin{proof} From (\ref{Hntd}), (\ref{dnt1}) and (\ref{hna}), we have
$$
H_n(t)=\sum_{j=0}^{n-1}\frac{d}{dt}\ln h_j(t)=\sum_{j=0}^{n-1}\alpha_j.
$$
In view of (\ref{suma}), we obtain (\ref{Hnt1}). By (\ref{pnt}), we arrive at (\ref{Hnt2}).
\end{proof}
The results obtained in this section will play an important role in the derivation of the large $n$ asymptotics of many quantities in the next section.

\section{Large $n$ asymptotics}
It is well-known that Hankel determinants are closely related to partition functions in random matrix theory \cite{Deift,Forrester2010,Mehta}. Let $Z_n(t)$ be the partition function for the unitary random matrix ensemble associated with the weight (\ref{airy}), i.e.,
$$
Z_n(t):=\int_{(0,\infty)^n}\prod_{1\leq i<j\leq n}(x_i-x_j)^2\prod_{k=1}^n x_k^\lambda\mathrm{e}^{-\frac{1}{3}x_k^3+tx_k}dx_k
$$
and the joint probability density function for the eigenvalues $x_1, x_2, \ldots, x_n$ of $n\times n$ Hermitian matrices from the ensemble are given by
$$
p(x_1, x_2, \ldots, x_n)=\frac{1}{Z_n(t)}\prod_{1\leq i<j\leq n}(x_i-x_j)^2\prod_{k=1}^n x_k^\lambda\mathrm{e}^{-\frac{1}{3}x_k^3+tx_k}.
$$
Then we have $D_n(t)=\frac{1}{n!}Z_n(t)$ \cite[(2.2.11)]{Szego}.

If we interpret the eigenvalues $x_1, x_2, \ldots, x_n$ as the positions of $n$ charged particles,
Dyson's Coulomb fluid approach {\cite{Dyson}} shows that in the limit of large $n$ the collection of particles can be approximated as a continuous fluid with a density $\sigma(\cdot)$ supported in $J$ (a subset of $\mathbb R^+$). When $\lambda\geq0$, the potential ${\rm v}(x)$ in (\ref{po}) is convex on $\mathbb R^+$. In this case, $J$ is a single interval, denoted by $(a, b),\; a>0$; see Chen and Ismail \cite{ChenIsmail} and also \cite[p. 198]{ST1997}.

According to \cite{ChenIsmail}, the equilibrium density $\sigma(\cdot)$ is obtained by minimizing the free energy functional
\be\label{fsigma}
F[\sigma]:=\int_a^b{\rm v}(x)\sigma(x)dx-\int_a^b \int_a^b \sigma(x)\ln|x-y|\sigma(y)dxdy
\ee
subject to the normalization condition
\be\label{nocon}
\int_a^b\sigma(x)dx=n.
\ee
Upon minimization, the density $\sigma(x)$ is found to satisfy the integral equation
\be\label{A}
{\rm v}(x)-2\int_a^b \ln|x-y|\sigma(y)dy=A,\qquad x\in(a,b),
\ee
where $A$ is the Lagrange multiplier for the constraint (\ref{nocon}). From (\ref{fsigma}), (\ref{nocon}) and (\ref{A}), it can be shown that \cite[(2.14)]{ChenIsmail}
\be\label{fn}
\frac{\partial F[\sigma]}{\partial n}=A.
\ee

Taking a derivative of (\ref{A}) with respect to $x$ gives
\be\label{sie}
{\rm v}'(x)-2P\int_a^b\frac{\sigma(y)}{x-y}dy=0,\qquad x\in(a,b),
\ee
where $P$ denotes the Cauchy principal value.
From (\ref{sie}) and the normalization condition (\ref{nocon}), it can be found that the endpoints $a$ and $b$ of the support of the density are determined by two supplementary conditions
\be\label{sc1}
\int_a^b\frac{{\rm v}'(x)}{\sqrt{(b-x)(x-a)}}dx=0,
\ee
\be\label{sc2}
\int_a^b\frac{x{\rm v}'(x)}{\sqrt{(b-x)(x-a)}}dx=2\pi n.
\ee
Furthermore, it was shown in {\cite{ChenIsmail}} that the recurrence coefficients have the following asymptotic behavior as $n\to\infty$,
\begin{subequations}\label{albe}
\be
\alpha_n=\frac{a+b}{2}+O\left(\frac{\partial^2A}{\partial t\partial n}\right),
\ee
\be
\beta_n=\left(\frac{b-a}{4}\right)^2\left(1+O\left(\frac{\partial^3A}{\partial n^3}\right)\right).
\ee
\end{subequations}

Substituting (\ref{po}) into (\ref{sc1}) and (\ref{sc2}), we obtain a system of equations:
$$
3X^2-4Y^2-8t-\frac{8\lambda}{Y}=0,
$$
$$
X(5X^2-12Y^2-8t)=16(2n+\la),
$$
where
$$
X=a+b,\qquad\qquad Y=\sqrt{ab}.
$$
This system has a unique solution subject to the conditions $X>0, Y>0$ and the large $n$ series expansions read
\begin{subequations}\label{XY}
\be
X=\frac{4n^{1/3}}{\kappa}+\frac{4 t}{3\ka^2n^{1/3}}+\frac{2\lambda}{3\ka n^{2/3}}-\frac{2 \lambda t}{9\kappa^2n^{4/3}}-\frac{2(t^3-90\lambda^2)}{405 \ka n^{5/3}}+\frac{ 2t(t^3+720\lambda^2)}{1215\ka^2n^{7/3}}+O(n^{-8/3}),\label{eqn:4.13a}
\ee
\be
Y=\frac{5\lambda}{3\ka n^{2/3}}+\frac{5\lambda t}{3\ka^2n^{4/3}}-\frac{5\lambda^2}{9\ka n^{5/3}}+\frac{4\lambda t^2}{27n^2}-\frac{10 \lambda^2 t}{9\ka^2n^{7/3}}+\frac{5\lambda(13t^3+15\lambda^2)}{486\ka n^{8/3}}-\frac{4\lambda^2 t^2}{27n^3}+O(n^{-10/3}),\label{eqn:4.13b}
\ee
\end{subequations}
where $\kappa=\sqrt[3]{10}$ and this notation will be used throughout the following text.
It follows that
\begin{subequations}\label{two}
\be
\frac{a+b}{2}=\frac{X}{2}=\frac{2n^{1/3}}{\kappa}+\frac{2 t}{3\ka^2n^{1/3}}+\frac{\lambda}{3\ka n^{2/3}}-\frac{ \lambda t}{9\kappa^2n^{4/3}}-\frac{t^3-90\lambda^2}{405 \ka n^{5/3}}+\frac{ t(t^3+720\lambda^2)}{1215\ka^2n^{7/3}}+O(n^{-8/3}),\label{eqn:4.14a}
\ee
\begin{align}\label{eqn:4.14b}
\left(\frac{b-a}{4}\right)^2=\frac{X^2-4Y^2}{16}=&\frac{n^{2/3}}{\kappa^2}+\frac{t}{15}+\frac{\lambda}{3\kappa^2n^{1/3}}+\frac{t^2}{90\kappa n^{2/3}}-\frac{t^3+180\lambda^2}{405\kappa^2n^{4/3}}-\frac{\lambda t^2}{270\kappa n^{5/3}}\no\\[8pt]
&-\frac{2\lambda^2 t}{27n^2}+\frac{\lambda(2t^3+345\lambda^2)}{1215\kappa^2n^{7/3}}+\frac{t^2(t^3-3600\lambda^2)}{36450\kappa n^{8/3}}+O(n^{-3}).
\end{align}
\end{subequations}

\begin{lemma}
The Lagrange multiplier $A$ has the following large $n$ asymptotic expansion:
\begin{align}\label{Aas}
A=&-\frac{2}{3}n\ln n+\frac{2}{3}(1+\ln10)n-\frac{2t n^{1/3}}{\kappa}-\frac{\lambda}{3}\ln n+\la\ln\ka-\frac{t^2}{3\kappa^2n^{1/3}}\nonumber\\[8pt]
&-\frac{\lambda t}{3\kappa n^{2/3}}-\frac{\lambda^2}{2n}+\frac{\lambda t^2}{18\kappa^2n^{4/3}}+\frac{t(t^3-360\lambda^2)}{1620\kappa n^{5/3}}+O(n^{-2}).
\end{align}
\end{lemma}

\begin{proof}
Similarly as in \cite[Lemma 3]{Min2021},
multiplying both sides of (\ref{A}) by $\frac{1}{\sqrt{(b-x)(x-a)}}$ and integrating with respect to $x$ over the interval $[a,b]$, we have
\begin{align}
A&=\frac{1}{\pi}\int_{a}^{b}\frac{{\rm v}(x)}{\sqrt{(b-x)(x-a)}}dx-2n\ln\frac{b-a}{4}\no\\
&=\frac{1}{48}(a+b)(5a^2-2ab+5b^2-24t)-\lambda\ln\frac{a+b+2\sqrt{ab}}{4}-2n\ln\frac{b-a}{4}\nonumber\\[8pt]
&=\frac{1}{48}X(5X^2-12Y^2-24t)-\lambda\ln\frac{X+2Y}{4}-n\ln\frac{X^2-4Y^2}{16}.\nonumber
\end{align}
Substituting (\ref{XY}) into the above and taking a large $n$ limit, we obtain the desired result.
\end{proof}
With these ingredients in hand, we are now ready to derive the large $n$ asymptotics of the recurrence coefficients.
\begin{theorem} For fixed $t\in \mathbb{R}$, the recurrence coefficients $\alpha_n$ and $\beta_n$ have the following asymptotic expansions as $n\rightarrow\infty$:
\begin{subequations}\label{albe2}
\begin{align}
\alpha_n=&\frac{2n^{1/3}}{\kappa}+\frac{2 t}{3\ka^2n^{1/3}}+\frac{\lambda+1}{3\kappa n^{2/3}}-\frac{(\lambda+1)t}{9\kappa^2n^{4/3}}-\frac{t^3-45(\la-1)(2\la+1)}{405\kappa n^{5/3}}\nonumber\\[8pt]
&+\frac{t\left[t^3+30(24\lambda^2+3\lambda-4)\right]}{1215\kappa^2n^{7/3}}+\frac{(\lambda+1)\left[t^3-15(7\lambda^2-\la-2)\right]}{486\kappa n^{8/3}}+O(n^{-3}),\label{eqn:4.18a}
\end{align}
\begin{align}
\beta_n=&\frac{n^{2/3}}{\kappa^2}+\frac{t}{15}+\frac{\lambda}{3\kappa^2n^{1/3}}+\frac{t^2}{90\kappa n^{2/3}}-\frac{t^3+180\lambda^2-45}{405\kappa^2n^{4/3}}-\frac{\lambda t^2}{270\kappa n^{5/3}}\no\\[8pt]
&-\frac{2t(15\lambda^2-4)}{405 n^2}+\frac{\lambda(2t^3+345\lambda^2-90)}{1215\kappa^2n^{7/3}}+\frac{t^2(t^3-3600\lambda^2+870)}{36450\kappa n^{8/3}}+O(n^{-3}).\label{eqn:4.18b}
\end{align}
\end{subequations}
\end{theorem}
\begin{proof}
From (\ref{albe}), (\ref{two}) and (\ref{Aas}), we see that $\alpha_n$ and $\beta_n$ have the following large $n$ expansion forms:
\begin{subequations}\label{albe1}
\be\label{an}
\alpha_n=\frac{2n^{1/3}}{\kappa}+a_0+\sum_{j=1}^{\infty}\frac{a_j}{n^{j/3}},
\ee
\be\label{bn}
\beta_n=\frac{n^{2/3}}{\kappa^2}+b_{-1}n^{1/3}+b_0+\sum_{j=1}^{\infty}\frac{b_j}{n^{j/3}}.
\ee
\end{subequations}
Substituting (\ref{albe1}) into the discrete system (\ref{dieq}) and letting $n\to\infty$, we obtain the expansion coefficients $a_j$ and $b_j$ recursively by equating powers of $n$ on both sides. The first few terms are
\begin{align}
&a_0=0,\qquad b_{-1}=0, \qquad a_1=\frac{2 t}{3\ka^2},\qquad b_0=\frac{t}{15},\qquad a_2=\frac{\lambda+1}{3\kappa},\qquad b_1=\frac{\lambda}{3\kappa^2},\nonumber\\[8pt]
&a_3=0,\qquad b_2=\frac{t^2}{90\kappa},\qquad a_4=-\frac{(\lambda+1)t}{9\kappa^2},\qquad b_3=0,\nonumber\\[8pt]
&a_5=-\frac{t^3-45(\la-1)(2\la+1)}{405\kappa n^{5/3}},\qquad b_4=-\frac{t^3+180\lambda^2-45}{405\kappa^2},\nonumber
\end{align}
and more terms are easily computable. This completes the proof.
\end{proof}
\begin{remark}
We have justified the large $n$ expansion forms of the recurrence coefficients in \cite[(46)]{Clarkson4}.
\end{remark}
\begin{remark}
Comparing (\ref{albe2}) with (\ref{two}) and taking account of (\ref{Aas}), one would find that the formulas in (\ref{albe}) are accurate.
\end{remark}
\begin{theorem} For fixed $t\in \mathbb{R}$, the sub-leading coefficient ${\rm\bf p}(n,t)$ has the following asymptotic expansion as $n\to{\infty}$:
\begin{align}\label{apn}
{\rm\bf p}(n,t)=&-\frac{3n^{4/3}}{2\kappa}-\frac{t n^{2/3}}{\kappa^2}-\frac{\lambda n^{1/3}}{\kappa}-\frac{t^2}{30}-\frac{\lambda t}{3\kappa^2n^{1/3}}-\frac{t^3-90\lambda^2 +15}{270\kappa n^{2/3}}\nonumber\\[8pt]
&+\frac{t(t^3+720\lambda^2-180)}{1620\kappa^2n^{4/3}}+\frac{\lambda(t^3-105\lambda^2 +15)}{810\kappa n^{5/3}}+O(n^{-2}).
\end{align}
\end{theorem}
\begin{proof} Substituting (\ref{albe2}) into (\ref{pnt}) and taking a large $n$ limit, we establish the theorem.
\end{proof}
Finally, we devote ourselves to the derivation of the large $n$ asymptotic expansion of the Hankel determinant $D_n(t)$ for fixed $t\in \mathbb{R}$.
By using a recent result of Charlier and Gharakhloo {\cite{CG2021}}, we first show the large $n$ asymptotic expansion (without higher order terms) of $D_n(0)$ in the following lemma.
\begin{lemma}
The Hankel determinant $D_n(0)$ has the large $n$ asymptotics
\begin{align}\label{logD0}
\ln D_n(0)=&\:\frac{1}{3}n^2\ln n-\left(\frac{1}{2}+\frac{\ln10}{3}\right)n^2+\frac{\lambda}{3}n\ln n+\left(\ln(2\pi)-\frac{\lambda}{3}(1+\ln 10)\right)n+\frac{3\lambda^2-1}{6}\ln n\nonumber\\[8pt]
&+2\zeta'(-1)-\ln G(\la+1)+\frac{\la}{2}\ln(2\pi)-\frac{\ln 3}{24}-\frac{4\la^2-1}{8}\ln\frac{5}{3}+ o(1),
\end{align}
where $\zeta(\cdot)$ is the Riemann zeta function and $G(\cdot)$ is the Barnes $G$-function.
\end{lemma}
\begin{proof}
Recall that
\begin{align}
D_n(0)&=\det\left(\int_0^{\infty}x^{i+j}x^{\lambda}\mathrm{e}^{-\frac{1}{3}x^3}dx\right)_{i,j=0}^{n-1}\nonumber\\
&=\frac{1}{n !}\int_{(0,\infty)^n}\prod_{1\le i<j\le n}(x_i-x_j)^2\prod_{k=1}^nx_k^{\lambda}\mathrm{e}^{-\frac{1}{3}x_k^3}dx_k.\no
\end{align}
In order to use the result in {\cite{CG2021}}, we make the change of variables $x_k=\sqrt[3]{\frac{4n}{5}}\:(y_k+1),\; k=1,2,\dots,n$ to obtain
\begin{align}\label{Dn0}
D_n(0)&=\left(\frac{4n}{5}\right)^{\frac{n(n+\lambda)}{3}}\frac{1}{n !}\int_{(-1,\infty)^n}\prod_{1\le i<j\le n}(y_i-y_j)^2\prod_{k=1}^n(y_k+1)^{\lambda}\mathrm{e}^{-\frac{4}{15}n(y_k+1)^3}dy_k\nonumber\\
&=\left(\frac{4n}{5}\right)^{\frac{n(n+\lambda)}{3}}\det\left(\int_{-1}^{\infty}x^{i+j}(x+1)^{\lambda}\mathrm{e}^{-\frac{4}{15}n(x+1)^3}dx\right)_{i,j=0}^{n-1}.
\end{align}
The Hankel determinant in (\ref{Dn0}) is a special case of \cite[Theorem 1.2]{CG2021} by taking
$$
V(x)=\frac{4}{15}(x+1)^3,\qquad W(x)=0,
$$
$$
\al_0=\la,\quad \al_1=\al_2=\cdots=\al_m=0,\qquad \bt_1=\bt_2=\cdots=\bt_m=0,
$$
$$
(\mathrm{do}\; \mathrm{not}\; \mathrm{confuse}\; \mathrm{these}\; \mathrm{notations}\; \mathrm{with}\; \mathrm{the}\; \mathrm{recurrence}\; \mathrm{coefficients})
$$
and
\begin{align}\label{psix}
\psi(x)&=\frac{1}{2\pi^2}P\int_{-1}^1\frac{V'(y)}{y-x}\sqrt{\frac{1+y}{1-y}}dy\no\\
&=\frac{2x^2+6x+7}{5\pi}.\nonumber
\end{align}
In this case, $\psi(x)>0$ for all $x\in[-1,1]$ and satisfies the normalization condition $\int_{-1}^{1}\psi(x)\sqrt{\frac{1-x}{1+x}}=1$.
Then we obtain from \cite[Theorem 1.2]{CG2021} that
\be\label{det}
\det\left(\int_{-1}^{\infty}x^{i+j}(x+1)^{\lambda}\mathrm{e}^{-\frac{4}{15}n(x+1)^3}dx\right)_{i,j=0}^{n-1}=\exp\left(C_1n^2+C_2n+C_3\ln n+C_4+o(1)\right),
\ee
where
\begin{align}
&C_1=-\frac{1}{2}-\ln 2,\qquad\qquad C_2=\ln(2\pi)-\lambda\left(\frac{1}{3}+\ln2\right),\qquad\qquad C_3=\frac{\lambda^2}{2}-\frac{1}{6},\nonumber\\
&C_4=2\zeta'(-1)-\ln G(\la+1)+\frac{\la}{2}\ln(2\pi)-\frac{\ln 3}{24}-\frac{4\la^2-1}{8}\ln\frac{5}{3}.\nonumber
\end{align}
The combination of (\ref{Dn0}) and (\ref{det}) gives the desired result.
\end{proof}
\begin{theorem}
For fixed $t\in \mathbb{R}$, the Hankel determinant $D_n(t)$ has the large $n$ asymptotic expansion
\begin{align}\label{dnta}
\ln D_n(t)=&\:\frac{1}{3}n^2\ln n-\left(\frac{1}{2}+\frac{\ln10}{3}\right)n^2+\frac{3tn^{4/3}}{2\kappa}+\frac{\lambda}{3}n\ln n+\left(\ln(2\pi)-\frac{\lambda}{3}(1+\ln10)\right)n\nonumber\\[8pt]
&+\frac{t^2n^{2/3}}{2\kappa^2}+\frac{\lambda t n^{1/3}}{\kappa}+\frac{3\lambda^2-1}{6}\ln n+c_0+\frac{\lambda t^2}{6\kappa^2n^{1/3}}+\frac{t(t^3-360\lambda^2 +60)}{1080\kappa n^{2/3}}\nonumber\\[8pt]
&+\frac{\lambda(8\lambda^2-3)}{36n}-\frac{t^2(t^3+1800\la^2-450)}{8100\kappa^2n^{4/3}}-\frac{\lambda t(t^3-420\lambda^2 +60)}{3240\kappa n^{5/3}}+O(n^{-2}),
\end{align}
where $\kappa=\sqrt[3]{10}$ and the constant term $c_0$ is given by
$$
c_0=\frac{t^3}{90}+2\zeta'(-1)-\ln G(\la+1)+\frac{\la}{2}\ln(2\pi)-\frac{\ln 3}{24}-\frac{4\la^2-1}{8}\ln\frac{5}{3}.
$$
\end{theorem}
\begin{proof}
From (\ref{fn}) and (\ref{Aas}), we find that the free energy $F[\sigma]$ has the large $n$ asymptotic expansion
\begin{align}
F[\sigma]=&-\frac{1}{3}n^2\ln n-\frac{\lambda}{3}n\ln n-\frac{\lambda^2}{2}\ln n+\left(\frac{1}{2}+\frac{\ln10}{3}\right)n^2-\frac{3tn^{4/3}}{2\kappa}+\frac{\lambda}{3}(1+\ln 10)n\nonumber\\[8pt]
&-\frac{t^2n^{2/3}}{2\kappa^2}-\frac{\lambda t n^{1/3}}{\kappa}+C-\frac{\lambda t^2}{6\kappa^2n^{1/3}}-\frac{t(t^3-360\lambda^2)}{1080\kappa n^{2/3}}-\frac{2\lambda^3}{9n}+\frac{t^2(t^3+1800\lambda^2)}{8100\kappa^2n^{4/3}}\nonumber\\[8pt]
&+\frac{\lambda t(t^3-420\lambda^2)}{3240\kappa n^{5/3}}+O(n^{-2}),\no
\end{align}
where $C$ is an integration constant.
It was pointed out in \cite{ChenIsmail,ChenIsmail1998} that the ``free energy" $F_n(t):=-\ln D_n(t)$ is approximated by the free energy $F[\s]$ for sufficiently large $n$ and the approximation is very accurate and effective. Hence, we have the large $n$ expansion form of $\ln D_n(t)$ as follows:
\be\label{dnt}
\ln D_n(t)=c_9n^2\ln n+c_8n\ln n+c_7\ln n+\sum_{j=-\infty}^{6}c_jn^{j/3}.
\ee
From (\ref{betadn}) we have
\be\label{rbd}
\ln \bt_n(t)=\ln D_{n+1}(t)+\ln D_{n-1}(t)-2\ln D_n(t).
\ee
Substituting (\ref{eqn:4.18b}) and (\ref{dnt}) into (\ref{rbd}) and taking the large $n$ asymptotic expansions, we obtain the coefficients $c_j$ (except $c_3$ and $c_0$) by comparing powers of $n$ on both sides. It follows that
\begin{align}\label{llDnt}
\ln D_n(t)=&\:\frac{1}{3}n^2\ln n+\frac{\lambda}{3}n\ln n+\frac{3\la^2-1}{6}\ln n-\left(\frac{1}{2}+\frac{\ln 10}{3}\right)n^2+\frac{3tn^{4/3}}{2\kappa}+c_3n\nonumber\\[8pt]
&+\frac{t^2n^{2/3}}{2\kappa^2}+\frac{\lambda t n^{1/3}}{\kappa}+c_0+\frac{\lambda t^2}{6\kappa^2n^{1/3}}+\frac{t(t^3-360\lambda^2 +60)}{1080\kappa n^{2/3}}+\frac{\lambda(8\lambda^2-3)}{36n}\nonumber\\[8pt]
&-\frac{t^2(t^3+1800\la^2-450)}{8100\kappa^2n^{4/3}}-\frac{\lambda t(t^3-420\lambda^2 +60)}{3240\kappa n^{5/3}}+O(n^{-2}).
\end{align}

Next, we devote ourselves to determining the two constants $c_3$ and $c_0$.
From (\ref{Hntd}) and (\ref{Hnt1}) we have
$$
\ln\frac{D_n(t)}{D_n(0)}=\int_0^t H_n(s)ds=-\int_0^t {\rm\bf p}(n,s)ds.
$$
By making use of (\ref{apn}) yields
\begin{align}\label{lnDn}
\ln\frac{D_n(t)}{D_n(0)}=&\frac{3tn^{4/3}}{2\kappa}+\frac{t^2n^{2/3}}{2\kappa^2}+\frac{\lambda t n^{1/3}}{\kappa}+\frac{t^3}{90}+\frac{\lambda t^2}{6\kappa^2n^{1/3}}+\frac{t(t^3-360\lambda^2 +60)}{1080\kappa n^{2/3}}\nonumber\\[8pt]
&-\frac{t^2(t^3+1800\la^2-450)}{8100\kappa^2n^{4/3}}-\frac{\lambda t(t^3-420\lambda^2 +60)}{3240\kappa n^{5/3}}+O(n^{-2}).
\end{align}
The sum of (\ref{logD0}) and (\ref{lnDn}) gives
\begin{align}\label{dnt2}
\ln D_n(t)=&\:\frac{1}{3}n^2\ln n-\left(\frac{1}{2}+\frac{\ln10}{3}\right)n^2+\frac{3tn^{4/3}}{2\kappa}+\frac{\lambda}{3}n\ln n+\left(\ln(2\pi)-\frac{\lambda}{3}(1+\ln10)\right)n\nonumber\\[8pt]
&+\frac{t^2n^{2/3}}{2\kappa^2}+\frac{\lambda t n^{1/3}}{\kappa}+\frac{3\lambda^2-1}{6}\ln n+\frac{t^3}{90}+2\zeta'(-1)-\ln G(\la+1)+\frac{\la}{2}\ln(2\pi)\no\\[8pt]
&-\frac{\ln 3}{24}-\frac{4\la^2-1}{8}\ln\frac{5}{3}+ o(1).
\end{align}
Comparing (\ref{llDnt}) and (\ref{dnt2}), we have
\begin{align}
&c_3=\ln(2\pi)-\frac{\lambda}{3}(1+\ln 10),\nonumber\\
&c_0=\frac{t^3}{90}+2\zeta'(-1)-\ln G(\la+1)+\frac{\la}{2}\ln(2\pi)-\frac{\ln 3}{24}-\frac{4\la^2-1}{8}\ln\frac{5}{3}.\nonumber
\end{align}
This completes the proof.
\end{proof}
\begin{corollary}
For fixed $t\in \mathbb{R}$, the normalized constant $h_n(t)$ has the large $n$ asymptotic expansion
\begin{align}
\ln h_n(t)=&\:\frac{2}{3}n\ln n-\frac{2}{3}(1+\ln 10)n+\frac{2tn^{1/3}}{\ka}+\frac{\la+1}{3}\ln n+\ln(2\pi)-(\la+1)\ln\ka+\frac{t^2}{3\ka^2n^{1/3}}\no\\[8pt]
&+\frac{(\la+1)t}{3\ka n^{2/3}}+\frac{9\la^2+3\la-1}{18n}-\frac{(\la+1)t^2}{18\ka^2n^{4/3}}-\frac{t^4-180(\la-1)(2\la+1)t}{1620\ka n^{5/3}}+O(n^{-2}).\no
\end{align}
\end{corollary}
\begin{proof}
From (\ref{dnt1}) we have
\be\label{hd}
\ln h_n(t)=\ln D_{n+1}(t)-\ln D_n(t).
\ee
Substituting (\ref{dnta}) into (\ref{hd}), we obtain the desired result by taking a large $n$ limit.
\end{proof}

\section{Long-time asymptotics}
In this section, we consider the asymptotics of our problem as $t\rightarrow\pm\infty$ for fixed $n\in\mathbb{N}$.
From the asymptotics of $\al_0(t)$ and $\bt_0(t)$ as $t\rightarrow\pm\infty$ and making use of the Toda system (\ref{ts}), Clarkson and Jordaan \cite{Clarkson4} obtained the long-time asymptotics of the recurrence coefficients in the following theorem.
\begin{theorem}
As $t\to+\infty$, the recurrence coefficients $\alpha_n(t)$ and $\beta_n(t)$ have the asymptotic expansions
\begin{subequations}\label{ab1}
\begin{align}
&\alpha_n(t)=\sqrt t-\frac{2n-2\lambda+1}{4t}-\frac{12 n^2+12n(1-4 \lambda )+12 \lambda ^2-24 \lambda +5}{32t^{5/2}}+O(t^{-4}),\\[8pt]
&\beta_n(t)=\frac{n}{2\sqrt t}+\frac{n(n-2\lambda)}{4t^2}+\frac{5n(4 n^2-24 n\lambda +12 \lambda ^2+1)}{64t^{7/2}}+O(t^{-5}).\label{bt1}
\end{align}
\end{subequations}
As $t\to-\infty$, the recurrence coefficients $\alpha_n(t)$ and $\beta_n(t)$ have the asymptotic expansions
\begin{subequations}\label{ab2}
\begin{align}
&\alpha_n(t)=-\frac{2n+\lambda+1}{t}-\frac{(2n+\lambda+1)\left[10 n^2+10 n(\lambda +1) +(\la+2)(\la+3)\right]}{t^4}+O(t^{-7}),\\[8pt]
&\beta_n(t)=\frac{n(n+\lambda)}{t^2}+\frac{4n(n+\lambda)(5n^2+5n\lambda+\lambda^2+1)}{t^5}+O(t^{-8}).\label{bt2}
\end{align}
\end{subequations}
\end{theorem}
\begin{proof}
See Clarkson and Jordaan \cite[Lemma 3.16]{Clarkson4}. We produce one more term for the asymptotics of each recurrence coefficient in the $t\to+\infty$ case by using their method. More higher order terms can be derived both for the $t\to+\infty$ and $t\to-\infty$ cases, but the expressions are very long.
\end{proof}
Based on the above theorem, we are able to derive the long-time asymptotics of the sub-leading coefficient ${\rm\bf p}(n,t)$, the Hankel determinant $D_n(t)$ and the normalized constant $h_n(t)$.
\begin{theorem}
As $t\to+\infty$, the sub-leading coefficient ${\rm\bf p}(n,t)$ has the asymptotic expansion
\be\label{p1}
{\rm\bf p}(n,t)=-n \sqrt{t}+\frac{n (n-2 \lambda )}{4 t}+\frac{n (4 n^2-24n\lambda +12 \lambda ^2+1)}{32t^{5/2}}+O(t^{-4}).
\ee
As $t\to-\infty$, the sub-leading coefficient ${\rm\bf p}(n,t)$ has the asymptotic expansion
\be\label{p2}
{\rm\bf p}(n,t)=\frac{n (n+\lambda )}{t}+\frac{n (n+\lambda )\left(5 n^2+5n\lambda +\lambda ^2+1\right)}{t^4}+O(t^{-7}).
\ee
\end{theorem}
\begin{proof}
Substituting (\ref{ab1}) into (\ref{pnt}) and letting $t\to+\infty$, we obtain (\ref{p1}). Similarly, substituting (\ref{ab2}) into (\ref{pnt}) and letting $t\to-\infty$, we arrive at (\ref{p2}).
\end{proof}
\begin{theorem}
As $t\to+\infty$, the Hankel determinant $D_n(t)$ has the asymptotic expansion
\be\label{d1}
\ln D_n(t)=\frac{2}{3}n t^{3/2}-\frac{n (n-2 \lambda )}{4}  \ln t+\widetilde{C}_1(n)+\frac{n(4 n^2-24n\lambda+12 \lambda ^2+1)}{48t^{3/2}}+O(t^{-3}),
\ee
where the constant term $\widetilde{C}_1(n)$, independent of $t$, is given by
$$
\widetilde{C}_1(n)=\frac{n}{2}\ln\pi-\frac{n(n-1)}{2}\ln2+\ln G(n+1),
$$
and $G(\cdot)$ is the Barnes $G$-function.\\
As $t\to-\infty$, the Hankel determinant $D_n(t)$ has the asymptotic expansion
\be\label{d2}
\ln D_n(t)=-n (n+\lambda ) \ln (-t)+\widetilde{C}_2(n)+\frac{n (n+\lambda ) (5 n^2+5 n\lambda+\lambda ^2+1)}{3 t^3}+O(t^{-6}),
\ee
where the constant term $\widetilde{C}_2(n)$ reads
$$
\widetilde{C}_2(n)=\ln\frac{G(n+1)G(n+\lambda+1)}{G(\lambda+1)}.
$$
\end{theorem}
\begin{proof}
From (\ref{Hntd}) and (\ref{Hnt1}) and in view of (\ref{p1}), we have as $t\to+\infty$
$$
\frac{d}{dt}\ln D_n(t)=n \sqrt{t}-\frac{n (n-2 \lambda )}{4 t}-\frac{n (4 n^2-24n\lambda +12 \lambda ^2+1)}{32t^{5/2}}+O(t^{-4}).
$$
It follows that
\be\label{lnd}
\ln D_n(t)=\frac{2}{3}n t^{3/2}-\frac{n (n-2 \lambda )}{4}  \ln t+\widetilde{C}_1(n)+\frac{n(4 n^2-24n\lambda+12 \lambda ^2+1)}{48t^{3/2}}+O(t^{-3}),
\ee
where $\widetilde{C}_1(n)$ is an integration constant, independent of $t$, to be determined.
From (\ref{bt1}) we find as $t\to+\infty$
\be\label{lnb}
\ln\beta_n(t)=-\frac{1}{2}\ln t+\ln\frac{n}{2}+\frac{n-2\lambda}{2t^{3/2}}+\frac{16n^2-104n\lambda +44\lambda^2+5}{32t^3}+O(t^{-{9/2}}).
\ee
Substituting (\ref{lnb}) and (\ref{lnd}) into (\ref{rbd}), and comparing the constant terms, we obtain
\be\label{de}
\widetilde{C}_1(n+1)+\widetilde{C}_1(n-1)-2\widetilde{C}_1(n)=\ln\frac{n}{2}.
\ee
Since
$$
D_0(t)=1, \qquad\qquad D_1(t)=\mu_0(t),
$$
and it was shown in \cite[p. 15]{Clarkson4} that as $t\to+\infty$
$$
\mu_0(t)=t^{\lambda/2-1/4}\sqrt{\pi}\exp\left(\frac{2}{3}t^{3/2}\right)\left[1+\frac{12\lambda^2-24\lambda+5}{48t^{3/2}}+O(t^{-3})\right],
$$
we have
\be\label{ic}
\widetilde{C}_1(0)=0,\qquad\qquad \widetilde{C}_1(1)=\frac{1}{2}\ln \pi.
\ee
The second-order difference equation (\ref{de}) with the initial conditions (\ref{ic}) has a unique solution given by
$$
\widetilde{C}_1(n)=\frac{n}{2}\ln\pi-\frac{n(n-1)}{2}\ln2+\ln G(n+1).
$$
Hence, we obtain (\ref{d1}).

On the other hand,
from (\ref{Hntd}), (\ref{Hnt1}) and (\ref{p2}) we have as $t\to-\infty$
$$
\frac{d}{dt}\ln D_n(t)=-\frac{n (n+\lambda )}{t}-\frac{n (n+\lambda )\left(5 n^2+5n\lambda +\lambda ^2+1\right)}{t^4}+O(t^{-7}).
$$
It follows that
\be\label{lnd2}
\ln D_n(t)=-n (n+\lambda ) \ln (-t)+\widetilde{C}_2(n)+\frac{n (n+\lambda ) (5 n^2+5 n\lambda+\lambda ^2+1)}{3 t^3}+O(t^{-6}),
\ee
where $\widetilde{C}_2(n)$ is an integration constant to be determined.
From (\ref{bt2}) we find as $t\to-\infty$
\be\label{lnb2}
\ln\beta_n(t)=-2\ln(- t)+\ln (n(n+\lambda))+\frac{4 (5 n^2+5 n\lambda+\lambda ^2+1)}{t^3}+O(t^{-{6}}).
\ee
Substituting (\ref{lnb2}) and (\ref{lnd2}) into (\ref{rbd}), and comparing the constant terms, we obtain
\be\label{de2}
\widetilde{C}_2(n+1)+\widetilde{C}_2(n-1)-2\widetilde{C}_2(n)=\ln (n(n+\lambda)).
\ee
It was shown in \cite[p. 16]{Clarkson4} that as $t\to-\infty$
$$
\mu_0(t)=\frac{\Gamma(\lambda+1)}{(-t)^{\lambda+1}}\left[1+\frac{(\lambda+1)(\lambda+2)(\lambda+3)}{3t^3}+O(t^{-6})\right],
$$
and we have
\be\label{ic2}
\widetilde{C}_2(0)=0,\qquad\qquad \widetilde{C}_2(1)=\ln \Gamma(\la+1).
\ee
The combination of (\ref{de2}) and (\ref{ic2}) shows that
$$
\widetilde{C}_2(n)=\ln\frac{G(n+1)G(n+\lambda+1)}{G(\lambda+1)}.
$$
Hence, we arrive at (\ref{d2}).
\end{proof}

\begin{corollary}
As $t\to+\infty$, the normalized constant $h_n(t)$ has the asymptotic expansion
$$
\ln h_n(t)=\frac{2}{3}t^{3/2}-\frac{1}{4} ( 2 n-2 \lambda+1) \ln t+\widehat{C}_1(n)+\frac{12 n^2+12n(1-4 \lambda )+12 \lambda ^2-24 \lambda +5}{48t^{3/2}}+O(t^{-3}),
$$
where the constant term $\widehat{C}_1(n)$ is given by
$$
\widehat{C}_1(n)=\frac{1}{2}\ln \pi-n\ln 2+\ln \Gamma(n+1).
$$
As $t\to-\infty$, the normalized constant $h_n(t)$ has the asymptotic expansion
$$
\ln h_n(t)=-(2 n+\lambda +1) \ln (-t)+\widehat{C}_2(n)+\frac{(2 n+\lambda +1)\left[10 n^2+10n( \lambda +1)+(\la+2)(\la+3)\right]}{3t^{3}}+O(t^{-6}),
$$
where the constant term $\widehat{C}_2(n)$ reads
$$
\widehat{C}_2(n)=\ln \left(\Gamma(n+1)\Gamma(n+\la+1)\right).
$$
\end{corollary}
\begin{proof}
The results are obtained by substituting (\ref{d1}) and (\ref{d2}) into the equality (\ref{hd}), respectively.
\end{proof}
\section{Conclusion}
In this paper, we have derived the ladder operator equations and compatibility conditions for orthogonal polynomials with respect to a general Laguerre-type weight. It is found that the definitions of the auxiliary functions $A_n$ and $B_n$ must be modified in contrast to Chen and Ismail \cite{ChenIsmail2005}. We have applied the ladder operator equations and compatibility conditions to study the generalized Airy polynomials. In addition, we have used Dyson's Coulomb fluid approach to investigate the large $n$ asymptotics for the recurrence coefficients of the polynomials. Our results complement the ones in \cite{Clarkson4}. We have also obtained some new results for the large $n$ asymptotics and long-time asymptotics of the sub-leading coefficient of the monic generalized Airy polynomials, the normalized constant of the polynomials and the associated Hankel determinant. It should be pointed out that all the asymptotic expansions obtained in this paper can be extended to any higher order.

\section*{Acknowledgments}
This work was partially supported by the National Natural Science Foundation of China under grant number 12001212, by the Fundamental Research Funds for the Central Universities under grant number ZQN-902 and by the Scientific Research Funds of Huaqiao University under grant number 17BS402.

\section*{Conflict of Interest}
The authors have no competing interests to declare that are relevant to the content of this article.
\section*{Data Availability Statement}
Data sharing not applicable to this article as no datasets were generated or analysed during the current study.

\end{document}